\documentclass[11pt]{article}
\usepackage{graphicx}
\usepackage{amssymb,amsmath,amsthm}
\usepackage{subcaption}
\usepackage{fullpage}
\usepackage{enumerate}
\usepackage[left]{lineno}
\newcommand{\comment}[1]{} %{\large\bf Comments \small\bf #1}}
\newcommand{\cltring}[2]{\mathcal{#1}\mathit{\left(#2\right)}}

\newtheorem{theorem}{Theorem}

\newtheorem{remark}[theorem]{Remark}
\newtheorem{corollary}[theorem]{Corollary}
\newtheorem{lemma}[theorem]{Lemma}
\newtheorem{notation}[theorem]{Notation}
\newtheorem{claim}[theorem]{Claim}
\newcommand{\badcltr}{(\infty,\infty)}
\newcommand{\headless}[1]{\mathcal{#1}\setminus\{h(\mathcal{#1})\}}
\newcommand{\head}[1]{h(\mathcal{#1})}

\begin{document}
%\date{}    

\title{A Polynomial Algorithm for Balanced Clustering via Graph Partitioning\thanks{This research has received funding from the projects COFLA2 (Junta de Andaluc\'ia, P12-TIC-1362) and GALGO (Spanish Ministry of Economy and Competitiveness and MTM2016-76272-R AEI/FEDER,UE).\newline
\parbox{.1\textwidth}{
\protect \includegraphics[trim=10cm 6cm 10cm 5cm,clip,scale=0.15]{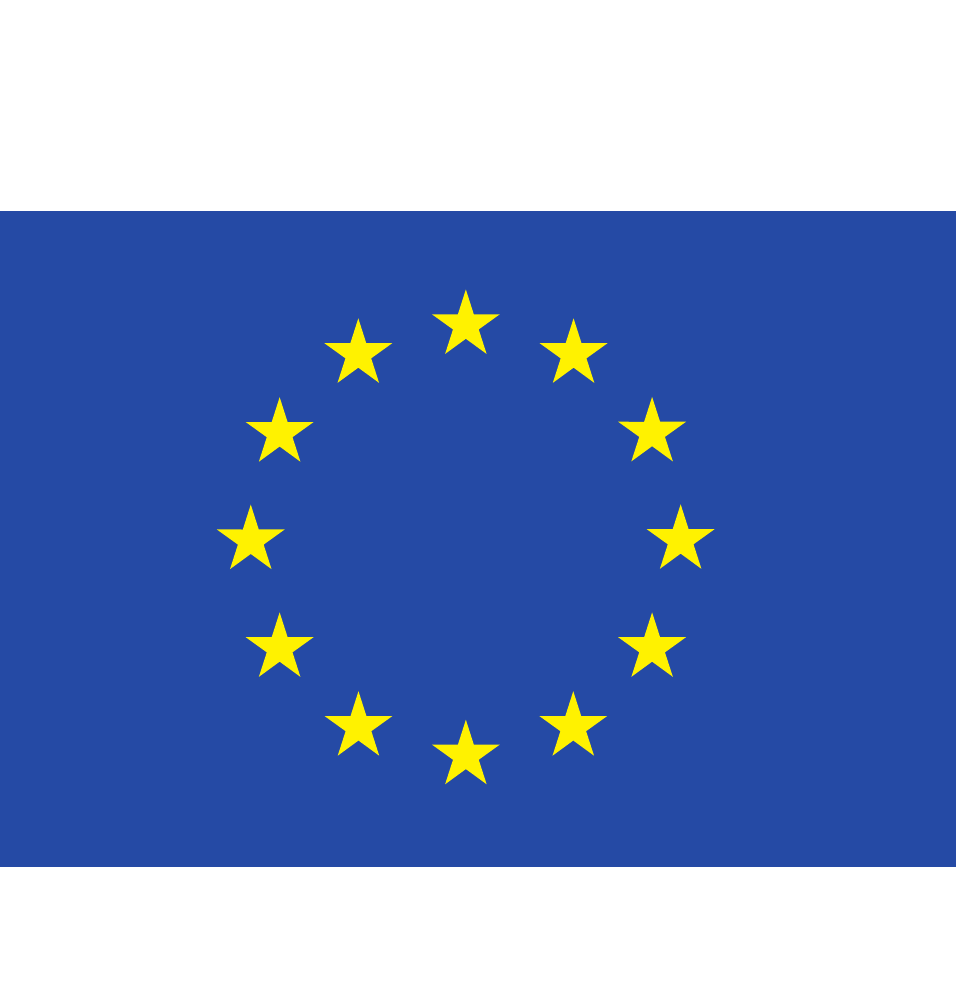}
}
\parbox{.9\textwidth}{ This work has also received funding from the European Union's Horizon 2020 research and innovation programme under the Marie Sk\l{}odowska-Curie grant agreement No 734922.}}
}

\author{Luis Evaristo Caraballo\thanks{Department of Applied Mathematics II, University of Seville, Spain. L.E.C. is funded by the Spanish Government under the FPU grant agreement FPU14/04705. Email: lcaraballo@us.es.}
\and
Jos\'e-Miguel D\'iaz-B\'a\~nez\thanks{Department of Applied Mathematics II, University of Seville, Spain. Email: dbanez@us.es.}
\and
Nadine Kroher\thanks{Department of Applied Mathematics II, University of Seville, Spain. Email: nkroher@us.es.}
}

\maketitle

\begin{abstract} 
The objective of clustering is to discover natural groups in datasets and to identify geometrical structures which might reside there, without assuming any prior knowledge on the characteristics of the data. The problem can be seen as detecting the inherent separations between groups of a given point set in a metric space governed by a similarity function. The pairwise similarities between all data objects form a weighted graph adjacency matrix which contains all necessary information for the clustering process, which can consequently be formulated as a graph partitioning problem. In this context, we propose a new cluster quality measure which uses the maximum spanning tree and allows us to compute the optimal clustering under the min-max principle in polynomial time. Our algorithm can be applied when a load-balanced clustering is required.
\end{abstract}

\section{Introduction}
The objective of clustering is to divide a given dataset into groups of similar objects in an unsupervised manner. Clustering techniques find frequent application in various areas, including computational biology, computer vision, data mining, gene expression analysis, text mining, social network analysis, VLSI design, and web indexing, to name just a few. Commonly, a metric is used to compute pair-wise similarities between all items and the clustering task is formulated as a graph partitioning problem, where a complete graph is generated from the similarity matrix. In fact, many graph-theoretical methods have been developed in the context of detecting and describing inherent cluster structures in arbitrary point sets using a distance function\cite{Asano1988}.

Here, we propose a novel clustering algorithm based on a quality measure that uses the maximum spanning tree of the underlying weighted graph and addresses a balanced grouping with the min-max principle. More specifically, we aim to detect clusters which are balanced with respect to their ratio of intra-cluster variance to their distance to other data instances. In other words, we allow clusters with a weaker inner edges to be formed, if they are located at large distance of other clusters (Figure \ref{fig:intro}). We prove that an optimal clustering under this measure can be computed in polynomial time using dynamic programming.

\begin{figure}
\centering
\includegraphics[scale=0.7]{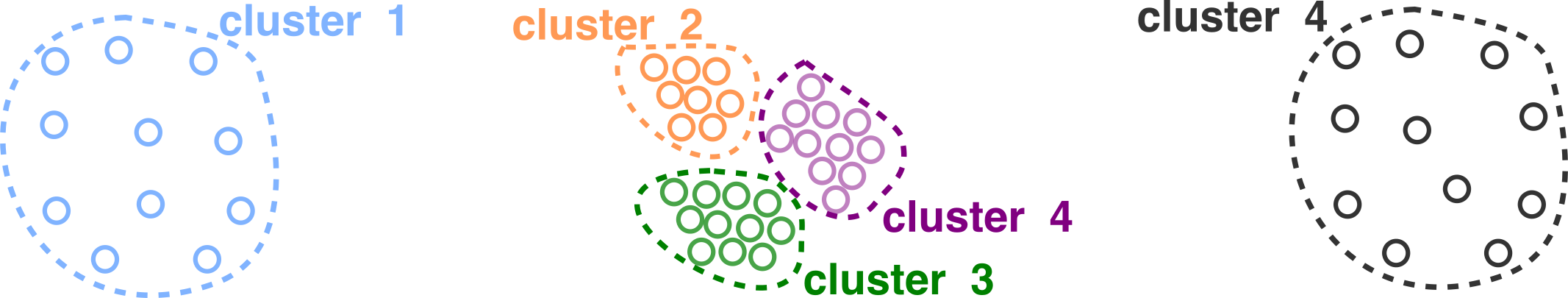}
\caption{Illustration of the desired cluster properties: The ratios of inner variance to distance to other clusters is balanced among groups. Clusters 1 and 4 exhibit a higher variance but are also further apart from the other clusters.}
\label{fig:intro}
\end{figure}

Such cluster properties are typically desired when grouping sensors in \textbf{wireless sensor networks} \cite{akyildiz2002wireless}: Each group communicates only with sensors in the same cluster and streams its information to a single command node located inside the cluster. The power consumption of sensors heavily depends on their distance to the command node and ideally, a balanced consumption among sensors is desirable. Consequently, a grouping should be balanced with respect to the ratio of inter-connection (sharing information between clusters) and intra-connection (sharing information within a cluster). A similar scenario occurs in the context of multi-robot task allocation in \textbf{cooperative robotics}, where the goal is to allocate tasks to robots while minimizing costs. For example, in monitoring missions, using a cooperative team of Unmanned Aerial Vehicles (UAVs), the goal is to minimize the elapsed time between two consecutive observations of any point in the area. Techniques based in area partitioning achieve this by assigning a sub-area to each UAV according to its capabilities. In this scenario, a load balanced clustering extends the life of the agents and allows to perform the task in a distributed manner \cite{ollero2007multiple, caraballo2017block}. Another possible application area arises from the field of \textbf{Music Information Retrieval}\cite{musicCluster}, where several applications rely on the unsupervised discovery of similar (but not identical) melodies or melodic fragments. In this context, clustering methods can be used to explore large music collections with respect to melodic similarity, or to detect repeated melodic patterns within a composition\cite{patterns}.  

 \subsection{Related work}
 Graph clustering refers to the task of partitioning a given graph $G(V,E)$ into a set of $k$ clusters $\mathcal{C} = \{C_1,\dots, C_k\}$ in such a way that vertices within a cluster are strongly connected whereas clusters are well separated. A number of exact and approximate algorithms have been proposed for this task, targeting different types of graphs (directed vs. undirected, complete vs. incomplete, etc.) and optimizing different cluster fitness values (i.e. maximizing densities, minimizing cuts). For a complete overview of existing strategies and their taxonomy, we refer to \cite{survey}. 
 The algorithm proposed in this study operates on the \textit{maximum spanning tree} $MST(G)$ of graph $G$. The idea of using minimum or maximum spanning trees (when working with distances or similarities, respectively) for cluster analysis goes back as far as 1971, when Zahn demonstrated \cite{Zahn1971} various properties which indicate that the minimum spanning tree serves as a suitable starting point for graph clustering algorithms and proposes a segmentation algorithm based on a local edge weight inconsistency criterion. This criterion was revisited and improved in \cite{Grygorash2006}. Asano et al. \cite{Asano1988} show that both, the optimal partitioning minimizing the maximum intra-cluster distances and the partitioning maximizing the minimum inter-cluster distance, can be computed from the maximum and minimum spanning trees. In the context of image processing, Xu et al. \cite{Xu1997} propose a dynamic programming algorithm for segmenting gray-level images which minimizes gray level variance in the resulting subtrees. Felzenszwalb et al. \cite{Felzenszwalb2004} introduced a comparison predicate which serves as evidence for a cluster boundary and provide a bottom-up clustering algorithm in $O(n\log{}n)$ time. In the context of Gene Expression Data Clustering, Xu et al. \cite{Xu2001} proposed three algorithms for partitioning the minimum spanning tree optimizing different quality criteria. 

\section{Problem statement}\label{sec:problemStatement}

Let $V=\{v_1,v_2,\dots,v_n\}$ be a set of points or nodes in a metric space and suppose that there exists a function to estimate the similarity between two nodes.  Let $A$ be the matrix of similarity computed for every pair of elements in $V$. The value $A[i,j]$ is the similarity between the nodes $v_i$ and $v_j$. If $A[i,j]>A[i,l]$, then the node $v_i$ is more similar to $v_j$ than $v_l$. 
Our goal is to create groups, such that similar nodes are in the same cluster and dissimilar nodes are in separate clusters.

Let $G=(V,E,w)$ be a weighted and undirected graph induced by the similarity matrix $A$ on the set of nodes $V$. In the sequel, such graphs are simply referred to as ``graph''. If $E$ is the set of edges and contains an edge for every unordered pair of nodes, and $w$ is a weight function $w:E\longrightarrow (0,1)$ such that $w(e)$ is the similarity between the nodes connected by $e$ (i.e. if $e=\{v_i,v_j\}$, then $w(e)=A[i,j]$).
%(if $w(\{v,v'\})>w(\{v,v''\})$ the melodies $v$ and $v'$ are more similar than $v$ and $v''$).% So, our problem can be analyzed as a graph clustering problem.

Let $C\subseteq V$ be a cluster. The \emph{outgoing edges set} of $C$, denoted by $Out(C)$ , is the set of edges connecting $C$ with $V\setminus C$. Let $MST(C)$ be the maximum spanning tree of $C$. Let $\max(Out(C))$ and $\min(MST(C))$ be the weights of the heaviest and lightest edges of $Out(C)$ and $MST(C)$, respectively.

We can use the following function as a quality measure of a cluster $C$:
$$\varPhi(C)=\left\lbrace \begin{array}{cl}
0 & \textnormal{if\quad} C = V,\\
&\\
\max(Out(C)) & \textnormal{if\quad} |C|=1,\\
&\\
\dfrac{\max(Out(C))}{\min(MST(C))} & \textnormal{in other case}
\end{array}\right.$$

Note that higher values of $\varPhi(\cdot)$ correspond to worse clusters. Also, note that, if $|C|=1$ then we can consider $\min(MST(C))=1$ and then  $\varPhi(C)=\max(Out(C))=\frac{\max(Out(C))}{\min(MST(C))}$. In addition, if $C=V$ then we can consider $\max(Out(C))=0$ and then  $\varPhi(C)=0=\frac{\max(Out(C))}{\min(MST(C))}$.

Let $\mathcal{C}=\{C_1\dots,C_k\}$ be a clustering formed by $k>1$ clusters of $G$.
To evaluate the quality of $\mathcal{C}$ we use the quality of the worst cluster of $\mathcal{C}$, that is, $\displaystyle\Phi(\mathcal{C})=\max_{i=1}^k\{\varPhi(C_i)\}$.

Denoting the set of all possible $k$-clusterings (clustering scenarios formed by $k$ clusters) on $G$ by $\cltring{P}{k,G}$, we state the following optimization problem (MinMax Clustering Problem):

\begin{equation}\label{opt_problem}
\begin{array}{rl}
\min &  \Phi(\mathcal{C})\\
\textnormal{subject to:} &
\mathcal{C}\in \cltring{P}{k,G}.
\end{array}\
\end{equation}

When the value of $k$ is unknown, the problem can be stated as follows:
\begin{equation}\label{opt_problem_k}
\begin{array}{rl}
\min &  \Phi(\mathcal{C})\\
\textnormal{subject to:} &
\mathcal{C}\in \displaystyle\bigcup_{k=2}^n\cltring{P}{k,G}.
\end{array}\
\end{equation}
That is, to find the clustering $\mathcal{C}\in\cltring{P}{k,G}$ such that $\Phi(\mathcal{C})$ is minimum among all possible clusterings with more than one cluster irrespective of the number of clusters contained in it.

\section{Properties of the optimal clustering}
Note that the problems stated above can be generalized to connected (not necessarily complete) graphs, by simply setting $\cltring{P}{k,G}$ as the set of all the possible partitions of $G$ in $k$ connected components. %and $\mathcal{C}^*\in\cltring{P}{k,G}$ is an optimal clustering (partition) of $G$ where every cluster is a connected component of $G$.
%Let $G$ be a graph of $n$ nodes and $k$ be a natural number such that $k\leq n$. Let $\mathcal{C^*}\in\cltring{P}{k,G}$ be an optimal clustering for the Problem~(\ref{opt_problem}), that is, $\Phi(\mathcal{C^*})\leq \Phi(\mathcal{C})$ for all clustering $\mathcal{C}\in \cltring{P}{k,G}$. 

\begin{lemma}\label{lemma1}
Let $G$ be a graph and let $\mathcal{C^*}$ be an optimal clustering of $G$ for Problem~{\ref{opt_problem}} in $\cltring{P}{k,G}$.
Then, $\Phi(\mathcal{C^*})\leq 1$.
\end{lemma}
\begin{proof}
If $k=1$, $\Phi(\mathcal{C^*})\leq 1$. If $k>1$, then take the maximum-spanning-tree of $G$ and denote it by $T=MST(G)$. Let $\mathcal{C}=\{C_1,\dots,C_k\}$ be the clustering induced by the $k$ connected components obtained by removing (or ``\emph{cutting}'') the $k-1$ lightest edges from $T$. If $|C_i|=1$ then $\varPhi(C_i)< 1$ since $w(e)\in(0,1)$ for all edge $e$ in the graph. If $|C_i|>1$, then $\min(MST(C_i))\geq \max(Out(C_i))$ (according to properties of a maximum-spanning-tree), therefore $\Phi(\mathcal{C})\leq 1$. The result follows.
\end{proof}

Before showing the next lemma, we recall the definition of the \emph{crossing edge}.
Let $G=(V,E)$ be a graph and let $\mathcal{B}=\{A,A'\}$ be a bipartition of $V$. An edge $\{v,v'\}\in E$ is a crossing edge of $\mathcal{B}$ if $v\in A$ and $v'\in A'$ (see Figure~\ref{fig:cross_edge}).

\begin{figure}
\centering
\includegraphics[scale=0.8]{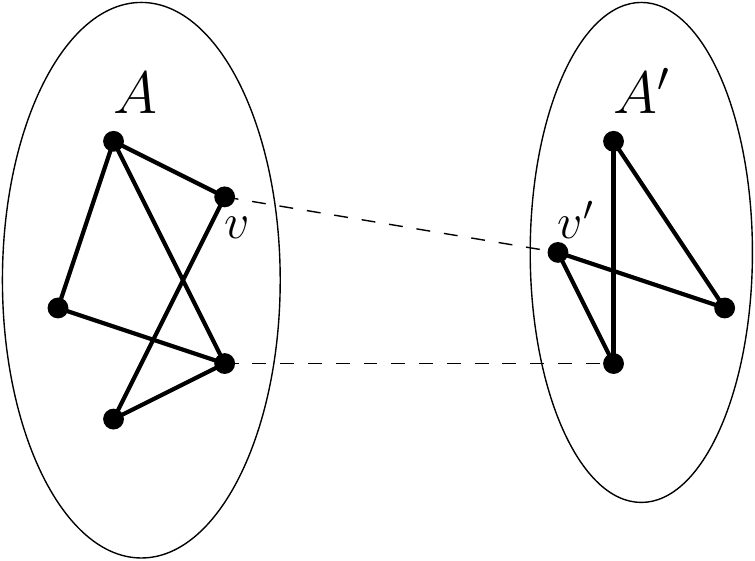}
\caption{Representation of a graph's bipartition. The crossing edges of this bipartition are indicated by dashed lines.}
\label{fig:cross_edge}
\end{figure}

\begin{lemma}\label{lemma2}
Let $G$ be a graph and let $\mathcal{C^*}$ be an optimal clustering of $G$ for Problem~{\ref{opt_problem}}. Let $C$ be a cluster in $\mathcal{C}^*$. If $|C|>1$, then every bipartition of $C$ has a crossing edge in a maximum spanning tree of $G$.
\end{lemma}
\begin{proof}
We prove by contradiction. Let $C$ be a cluster in $\mathcal{C}^*$ with cardinality greater than 1 and let $\{A,A'\}$ be a partition of $C$ such that there is no edge from $A$ to $A'$ in a maximum spanning tree of $G$. Let $e$ be the heaviest edge that crosses from $A$ to $A'$. Let $MST(G)$ be a spanning tree of $G$. Due to our assumption, $e\notin MST(G)$ and therefore, adding $e$ to $MST(G)$ results in a cycle. All other edges in this cycle have a weight equal or greater than $e$, given by the properties of the maximum spanning tree. Let $e'$ be an edge in this cycle connecting a node in $C$ with others in $V\setminus C$. If $w(e)=w(e')$, then replacing $e'$ by $e$ in $MST(G)$ we obtain another maximum spanning tree containing $e$ and thus, this is a contradiction. 
If $w(e)<w(e')$, then $\max(Out(C))\geq w(e')>w(e)\geq\min(MST(C))$.% because $e'\in Out(C)$. 
This is a contradiction by Lemma~\ref{lemma1}.
\end{proof}

\begin{theorem}\label{first-theorem}
Let $G$ be a graph and let $\mathcal{C^*}\in\cltring{P}{k,G}$ be an optimal clustering of $G$ for Problem~{\ref{opt_problem}}. For every cluster $C\in\mathcal{C^*}$, the maximum spanning tree of $C$ is a subtree of a maximum spanning tree of $G$ and the heaviest outgoing edge of $C$ is in a maximum spanning tree of $G$.
\end{theorem}
\begin{proof}
Let $C$ be a cluster of $\mathcal{C^*}$. If $|C|=1$ then, obviously, $MST(C)\subset MST(G)$. If $|C|>1$ then $MST(C)\subseteq MST(G)$ by Lemma~\ref{lemma2}. The second part of the theorem, claiming that the heaviest edge in $Out(C)$ is in $MST(G)$, is deduced from the properties of the maximum spanning tree of a graph.
\end{proof}

The following result is directly deduced from the theorem above.
\begin{corollary}
Let $G$ be a graph and let $\mathcal{C^*}\in\bigcup_{k=1}^n\cltring{P}{k,G}$ be an optimal clustering of $G$ for Problem~{\ref{opt_problem_k}}. For every cluster $C\in\mathcal{C^*}$, the maximum spanning tree of $C$ is a subtree of a maximum spanning tree of $G$ and the heaviest outgoing edge of $C$ is in a maximum spanning tree of $G$.
\end{corollary}

\begin{figure}[h]
\centering
\begin{subfigure}{.48\textwidth}
	\centering
 \includegraphics[width=0.6\textwidth]{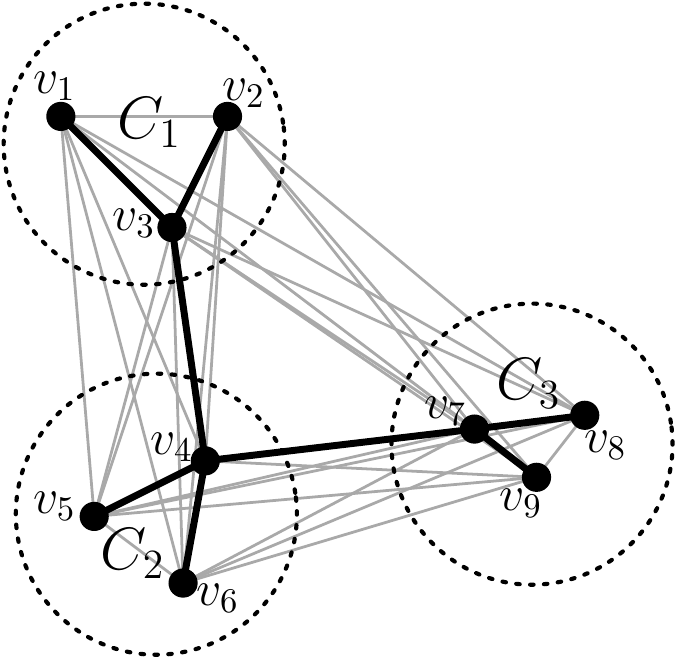}
 \caption{}
 \label{fig:graph_tree}
\end{subfigure}
\begin{subfigure}{.48\textwidth}
	\centering
	\includegraphics[width=0.6\textwidth, page = 2]{cutting-MST.pdf}
	\caption{}
	\label{fig:cutting_edges}
\end{subfigure}
\caption{(a) A graph $G$ and a spanning tree $T$ of $G$. The edges of $T$ are bold. The clustering $\mathcal{C}=\{C_1,C_2,C_3\}$ is represented by dotted strokes. (b) Obtaining an optimal clustering $\mathcal{C}^*=\{C_1,C_2,C_3\}$ by cutting two edges in the maximum spanning tree.}
\end{figure}

Let $G$ be a graph and let $T$ be a spanning tree of $G$. Note that every possible clustering of $T$ is a valid clustering in $G$ and therefore
$\cltring{P}{k,T}\subseteq\cltring{P}{k,G}$ (see Figure~\ref{fig:graph_tree}). 
However, a valid clustering of $G$ may be not feasible for $T$, for example, the clustering $\{\{v_1,v_3,v_5\},\{v_4,v_6\},\{v_2,v_7,v_8,v_9\}\}$ is valid for the graph $G$ in Figure~\ref{fig:graph_tree}, but is not feasible for $T$ because the cluster $\{v_1,v_3,v_5\}$ does not constitute a connected component in $T$. Let $C$ be a valid cluster for $G$ and $T$. Consider $Out_T(C)$ as the set of outgoing edges of $C$ as described earlier, but restricted to the set of edges forming $T$. For example, considering $G$ and $T$ in Figure~\ref{fig:graph_tree}, the set $Out_T(C_1)$ only contains the edge $\{v_3,v_4\}$; however, the set $Out(C_1)$ contains $\{v_3,v_4\},\{v_1,v_5\},\{v_1,v_6\},\dots$ Analogously, we can apply the same argument to the set of inner edges in the cluster $C$. We use the analogous notations $MST_T(\cdot)$ to denote the maximum spanning tree of a cluster using only the edges in $T$; also, note that $MST_T(C)$ is the subtree of $T$ determined by the nodes of $C$. Consequently $MST_G(C)=MST(C)$ and $Out_G(C)=Out(C)$. 

The previous explanation is needed to introduce the following notions:
Let $T$ be a spanning tree of a graph $G$. Let $\mathcal{C}\in\cltring{P}{k,T}$ be a clustering of $T$. The evaluation function $\Phi_T(\mathcal{C})$ operates as usual, but is restricted to the set of edges forming $T$. Therefore, the optimal solution for Problem~\ref{opt_problem} on $T$ is $\mathcal{C}^\dagger\in\cltring{P}{k,T}$ such that $\Phi_T(\mathcal{C}^\dagger)\leq \Phi_T(\mathcal{C})$ for every other clustering $\mathcal{C}\in\cltring{P}{k,T}$.

\begin{theorem}\label{main-theorem}
Let $G$ be a graph and let $T$ be a maximum spanning tree of $G$. If $\mathcal{C}^*\in\cltring{P}{k,G}$ and $\mathcal{C}^\dagger\in \cltring{P}{k, T}$ are the optimal clusterings (for Problem~{\ref{opt_problem}}) on $G$ and $T$, respectively; then $\Phi(\mathcal{C}^*)=\Phi_T(\mathcal{C}^\dagger)$. 
\end{theorem}
\begin{proof}
Let us first prove that $\Phi(\mathcal{C}^*)\geq\Phi_T(\mathcal{C}^\dagger)$.
Suppose that every cluster $C\in \mathcal{C}^*$ contains a single connected component in $T$. Then, for every cluster $C\in \mathcal{C}^*$, $\min(MST_G(C))=\min(MST_T(C))$ follows from properties of a maximum spanning tree. Moreover $\max(Out_G(C))\geq\max(Out_T(C))$ (because $T$ is a subgraph of $G$). Therefore, $\varPhi(C)\geq \varPhi_T(C)$, which implies $\Phi(\mathcal{C}^*)\geq \Phi_T(\mathcal{C}^*)\geq \Phi_T(\mathcal{C}^\dagger)$.

Suppose now, that some cluster $C\in \mathcal{C}^*$ does not contain a single connected component in $T$. For an illustration, see Figure~\ref{fig:proof1} where the cluster $C$ encompasses parts of two connected components in $T$. Take an edge $e\in MST(C)$ connecting nodes in two of the different connected components determined by $C$ in $T$. Adding $e$ to $T$ will form a cycle in the edges of $T$. If there exists an edge $e'$ such that $w(e)>w(e')$ in this cycle, then $T$ is not a maximum spanning tree. Consequently, for all edges $e'\neq e$ in this cycle $w(e')\geq w(e)$. Note that this cycle includes edges in $Out(C)$ (if all edges of this cycle were inner edges of $C$, then $e$ would connect nodes within the same connected component, posing a contradiction), then $\max(Out(C))\geq w(e)\geq \min(MST(C))$. Thus $1\leq \varPhi(C)\leq \varPhi(\mathcal{C}^*)$. On the other hand, $\Phi_T(\mathcal{C}^\dagger)\leq 1$ by using Lemma~\ref{lemma1}, so, $\Phi_T(\mathcal{C}^\dagger)\leq\varPhi(\mathcal{C}^*)$.

Finally, let us prove that $\Phi(\mathcal{C}^*)\leq\Phi_T(\mathcal{C}^\dagger)$.
For every cluster $C\in\mathcal{C}^\dagger$, $\min(MST_G(C))=\min(MST_T(C))$ and $\max(Out_G(C))=\max(Out_T(G))$ can be demonstrated using the properties of the maximum spanning tree. Therefore $\varPhi(C)=\varPhi_T(C)$, which implies that $\Phi(\mathcal{C}^\dagger)=\Phi_T(\mathcal{C}^\dagger)$. Then $\Phi(\mathcal{C}^*)\leq \Phi(\mathcal{C}^\dagger)$m because $\mathcal{C}^*$ is the optimal clustering on $G$. This completes the proof.
\end{proof}

\begin{figure}
\centering
\begin{subfigure}{.49\textwidth}
\centering
\includegraphics[page=1]{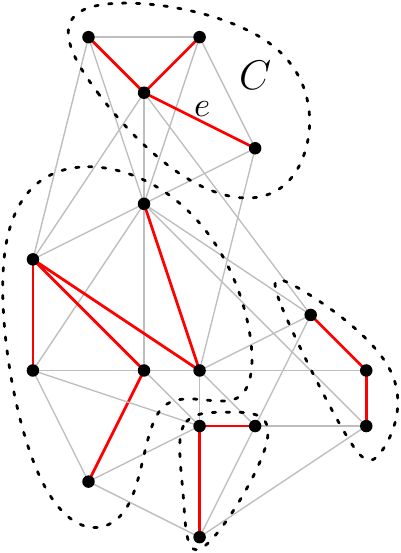}
\caption{}
\end{subfigure}
\begin{subfigure}{.49\textwidth}
\centering
\includegraphics[page=2]{proof1.pdf}
\caption{}
\end{subfigure}
\caption{Representation of a graph $G$ and a clustering $\mathcal{C}^*$. Dotted black strokes mark clustering boundaries. (a) The maximum spanning trees of the clusters are drawn in red. (b) A maximum spanning tree $T$ of $G$ is drawn in blue. The two connected components partially covered by cluster $C$ in $T$ are shaded gray.}
\label{fig:proof1}
\end{figure}

The following result is directly deduced from the theorem above.
\begin{corollary}\label{main-corollary}
Let $G$ be a graph and let $T$ be a maximum spanning tree of $G$. If $\mathcal{C}^*\in\bigcup_{k=1}^n\cltring{P}{k,G}$ and $\mathcal{C}^\dagger\in \bigcup_{k=1}^n\cltring{P}{k, T}$ are the optimal clusterings (for Problem~{\ref{opt_problem_k}}) on $G$ and $T$, respectively; then $\Phi(\mathcal{C}^*)=\Phi_T(\mathcal{C}^\dagger)$. 
\end{corollary}

As a consequence of the above properties, the optimal $k$-clustering $\mathcal{C^*}$ for Problem~\ref{opt_problem} can be obtained by ``cutting'' the $k-1$ appropriate edges in a maximum spanning tree of $G$ (see Figure~\ref{fig:cutting_edges}). These edges can be found combinatorially in $O(n^{k-1})$ time. Thus, using a naive approach, the solution of Problem~\ref{opt_problem_k} can be found in $O(n^{n-1})$ time. In the next section, we show an algorithm which solves both problems in polynomial time in $n$ and $k$.

\section{The algorithm}

First, recall that Theorem~\ref{main-theorem} and Corollary~\ref{main-corollary} provide a nice property, which allows us to reduce Problems ~\ref{opt_problem} and \ref{opt_problem_k} from a graph to its maximum spanning tree. Consequently, given a similarity graph, we can operate on its maximum spanning tree $T=(V,E,w)$. From now on, we will use $E$ to denote the set of edges in the maximum spanning tree. Observe that every cluster $C$ in $T$ determines only one subtree of $T$. Then, using $MST(C)$ to denote the maximum spanning tree in $C$ may be confusing or redundant. Therefore, instead of using $MST(C)$, we will use $E(C)$ (set of edges connecting nodes in $C$).

%With this notation we are ready to proof 
 
The following technical lemma is crucial for the correctness of our algorithm. 
\begin{lemma}\label{division_lemma}
	Let $\mathcal{C}$ be a clustering of a tree $T=(V,E,w)$. By removing an edge of $T$ we induce two clusterings, one for each generated subtree (see Figure~\ref{fig:div_cluster}). The evaluations of the induced clusterings are at most $\Phi(\mathcal{C})$.
\end{lemma}
\begin{proof}
Let $\mathcal{A}$ and $\mathcal{B}$ denote the two induced clusterings and let $e$ be the removed edge (see Figure~\ref{fig:div_cluster}). For the sake of contradiction, suppose that one of the two induced clusterings has an evaluation greater than $\mathcal{C}$. W.l.o.g. assume that $\Phi(\mathcal{A})>\Phi(\mathcal{C})$. Let $A$ denote the cluster of $\mathcal{A}$ containing one of the incident nodes of $e$ (the other one is in a cluster of $\mathcal{B}$).

If $\Phi(\mathcal{A})>\varPhi(A)$, then there is another cluster $A'\in\mathcal{A}$ such that $\varPhi(A')=\Phi(\mathcal{A})$. Note that $A'$ is also in $\mathcal{C}$ and $A'$ is not affected when $e$ is removed, so, $\Phi(\mathcal{C})\geq \varPhi(A')$. This is a contradiction since we are assuming that $\varPhi(A')=\Phi(\mathcal{A})>\Phi(\mathcal{C})$.

If $\Phi(\mathcal{A})=\varPhi(A)$, then let $m_A=\min(E(A))$ (recall that we consider $m_A=1$ if $|A|=1$) and $M_A=\max(Out(A))$. Observe that the incident nodes of $e$ may be both, in the same cluster of $\mathcal{C}$ (see Figure~\ref{edge_inside}), or not (see Figure~\ref{edge_out}). Suppose the incident nodes of $e$ are in the same cluster $C\in \mathcal{C}$. Note that $\min(E(C))\leq m_A$ and $\max(E(C))\geq M_A$, therefore $$\Phi(\mathcal{C})\geq\varPhi(C)=\frac{\max(Out(C))}{\min(E(C))}\geq \frac{M_A}{m_A}=\varPhi(A)=\Phi(\mathcal{A}).$$
This is another contradiction. 

Now, suppose that the incident nodes of $e$ are in different clusters of $\mathcal{C}$. In this case $\mathcal{A}\subset\mathcal{C}$ and then $A$ is also in $\mathcal{C}$. Let $M_A^{(\mathcal{C})}$ denote the weight of the heaviest outgoing edge of $A$ in $\mathcal{C}$, then $M_A^{(\mathcal{C})}=\max\{w(e),M_A\}\geq M_A$, and the evaluation of $A$ in $\mathcal{C}$ is $\varPhi^{(\mathcal{C})}(A)=\frac{M_A^{(\mathcal{C})}}{m_A}$, therefore
$$\Phi(\mathcal{C})\geq\varPhi^{(\mathcal{C})}(A)=\frac{M_A^{(\mathcal{C})}}{m_A}\geq \frac{M_A}{m_A}=\varPhi(A)=\Phi(\mathcal{A}).$$
This is a contradiction and completes the proof.
\end{proof}

\begin{figure}
	\centering
	\begin{subfigure}{\textwidth}
		\centering
		\includegraphics[scale=1.2]{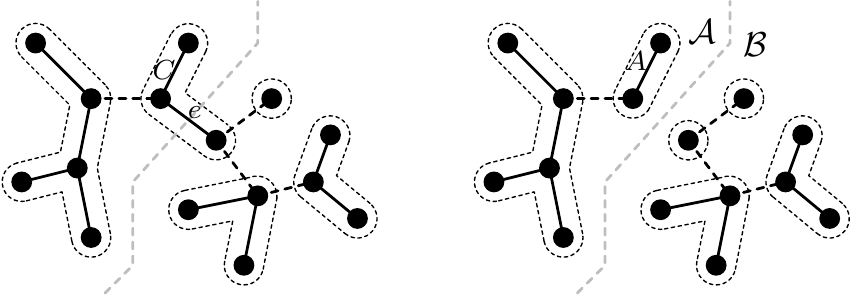}
		\caption{}
		\label{edge_inside}
	\end{subfigure}\vspace{.5cm}
	\begin{subfigure}{\textwidth}
		\centering
		\includegraphics[scale=1.2, page =2]{divide.pdf}
		\caption{}
		\label{edge_out}
	\end{subfigure}
	\caption{Obtaining two clusterings, one per subtree, by removing an edge of a given clustering.  The removed edge in (a) is inside a cluster, and in (b) is a crossing edge.}
	\label{fig:div_cluster}
\end{figure}

The proposed algorithm is based on \emph{dynamic programming}. We show, that the stated problems have an \emph{optimal substructure} and construct the optimal solution in $T$ from optimal solutions for subtrees of $T$. From here on, we consider that the tree $T$ is rooted at an arbitrary node $r\in V$. For all $v\in V$, let $c(v)$ be the set of children of $v$; and for all $v\in V, v\neq r$, let $p(v)$ be the parent of $v$. Recall that if $c(v)$ is empty then we say that $v$ is a \emph{leaf} node.

Given a tree $T$, let $S$ be a subtree of $T$ and let $v$ be the node with minimum depth in $S$. Then we say that $S$ is rooted at $v$. In the sequel, we only consider subtrees $S$ rooted at $v$ that contain all the descendants of vertices $v'\in S\setminus\{v\}$. 
% with following property: If $v'$ is in the subtree $S$ and $v'$ is different from $v$ (root of $S$), then every leaf of $T$ hanged from $v'$ is in $S$ too.
Figure~\ref{subt_invalid} shows a subtree $S'$ rooted at $v$. The leaves hanging from $v'$ and $v''$ are not in $S'$, so $S'$ is not considered as a subtree. Figure~\ref{subt_valid} shows an example of a subtree to be considered. In addition, we say that $S=T_v$ if $S$ is rooted at $v$ and contains all the descendants of $v$.

% Let $S$ be a subtree rooted at $v$, if $S$ contains all the children of $v$ then $S=T_v$, that is, $S$ contains all the nodes hanging from $v$. Note that $T_r=T$. 

\begin{figure}
\centering
\begin{subfigure}{.32\textwidth}
\centering
\includegraphics[page=3]{T_v_tree.pdf}
\caption{}
\label{subt_invalid}
\end{subfigure}
\begin{subfigure}{.32\textwidth}
	\centering
	\includegraphics[page=2]{T_v_tree.pdf}
	\caption{}
	\label{subt_valid}
\end{subfigure}
\begin{subfigure}{.32\textwidth}
	\centering
	\includegraphics[page=1]{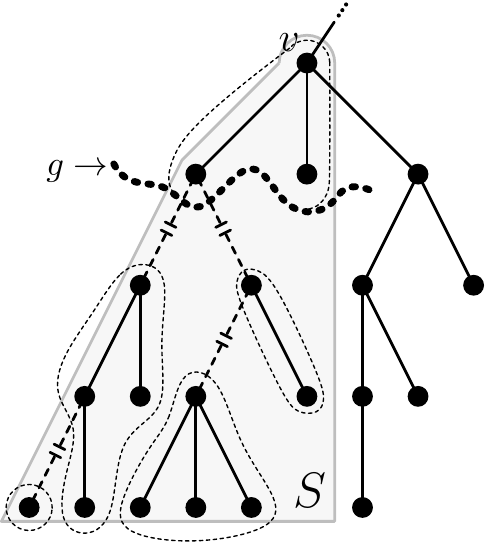}
	\caption{}
	\label{fig:T_v_sample}
\end{subfigure}
\caption{(a) A subtree $S'$ which is not considered. (b) A considered subtree $S$. (c) Representation of a clustering $\mathcal{C}$ of $S$. The \emph{head cluster} is above the curve $g$. The edges of $Out_S(h(\mathcal{C}))$ are the edges in $S$ crossed by $g$. The clusters of $\mathcal{C}$ that are below $g$ constitute the headless clustering $\mathcal{C}\setminus\{h(\mathcal{C})\}$.}
\end{figure}

The main idea of our algorithm is to work on (local) clusterings of a subtree and perform a bottom-up dynamic programming strategy with two basic operations:
\begin{itemize}
\item \textbf{UpToParent}: knowing an optimal clustering of the subtree $S=T_v\neq T$, compute an optimal clustering of the subtree formed by adding $p(v)$ to $S$, see Figure~\ref{up2parent}.

\item \textbf{AddChildTree}: knowing an optimal clustering of a subtree $S$ rooted at $p(v)$ and knowing an optimal clustering of the subtree $Q=T_v$ such that $v\not\in S$, compute an optimal clustering of the subtree formed by joining $S$ and $Q$, see Figure~\ref{merging}.
\end{itemize}

Now, we elaborate on a the (local) clustering $\mathcal{C}$ of a subtree $S$ rooted at $v$, see Figure~\ref{fig:T_v_sample}. 
A clustering of $S$ is given by cutting some edges.
%Focus on a subtree $S$ rooted at a node $v$ and consider a clustering $\mathcal{C}$ of $S$. %Now we are going to introduce the notions of \emph{head cluster} and \emph{descending outgoing edges} of a cluster. 
We call a cluster containing the node $v$ \emph{head cluster} of $\mathcal{C}$, denoted $h(\mathcal{C})$ (see Figure~\ref{fig:T_v_sample}). Note that if $C\in\mathcal{C}$ and $C\neq h(\mathcal{C})$, then $Out(C)$ is entirely contained in $S$. However, $Out(h(\mathcal{C}))$ is entirely contained in $S$ only if $S=T$; if $v\neq r$ then $\{v,p(v)\}\in Out(h(\mathcal{C}))$ and $\{v,p(v)\}$ is not in $S$; if $v=r$ and $S\neq T$, then some node $v'\in c(v)$ is not in $S$ and $\{v,v'\}\in Out(h(\mathcal{C}))$ but $\{v,v'\}$ is not in $S$. Thus, it is convenient to introduce $Out_S(C)$ as the set of outgoing edges of $C$ connecting nodes in $S$. In Figure~\ref{fig:T_v_sample}, $Out_S(h(\mathcal{C}))$ is formed by the edges stabbed by the curve $g$.

%Consider the set $D=w(E_T)\cup\{0,1\}$ which is the union set of the costs of the edges in $T$ and the values 0 and 1. Also, for convenience, consider that $\min(MST(C))=1$ if $C$ is a cluster formed by a single node ($|C|=1$, the cluster has no edges). % and, if $\mathcal{C}=\{C\}$ (that is $\mathcal{C}$ is formed by a single cluster containing all the nodes of $T_v$) then $\max(Out(C))=0$.

Given a clustering $\mathcal{C}$ of a subtree $S$, let $M$ be the weight of the heaviest edge in $Out_S(h(\mathcal{C}))$, that is, $M=\max(Out_S(h(\mathcal{C})))$. If $h(\mathcal{C})$ contains all the nodes in $S$, then there are no descending outgoing edges, and in these cases we set $M=0$. On the other hand, let $\mu$ be the weight of the lightest edge in $E(h(\mathcal{C}))$, that is $\mu=\min(E(h(\mathcal{C})))$. If $h(\mathcal{C})$ is formed by single node, that is $h(\mathcal{C})=\{v\}$, then $E(h(\mathcal{C}))$ is empty, and in these cases we set $\mu=1$. For convenience, we introduce the functions $\varPhi_{S}(\cdot)$ and $\Phi_{S}(\cdot)$ as the restricted quality measures of a cluster and a clustering, respectively. They work as usual but are restricted to the edges of the subtree $S$, thus:
\begin{equation}\label{head_measure}
\varPhi_{S}(h(\mathcal{C}))=\frac{M}{\mu}.
\end{equation}
Note that if $S=T$, then $\varPhi_{S}(h(\mathcal{C}))=\varPhi(h(\mathcal{C}))$. If $S=T_v\neq T$, then:
$$\varPhi(h(\mathcal{C}))=\frac{\max\{M,w(\{v,p(v)\}) \}}{\mu}.$$
For every cluster $C\in\mathcal{C}$, such that $C$ is not the head cluster, the usual evaluation and the restricted one have the same value, $\varPhi(C)=\varPhi_{S}(C)$. Consequently, the restricted evaluation of the ``\emph{headless}'' clustering $\mathcal{C}\setminus\{h(\mathcal{C})\}$ is:

\begin{equation}\label{headless_measure}
\Phi_{S}(\mathcal{C}\setminus\{h(\mathcal{C})\})=\Phi(\mathcal{C}\setminus\{h(\mathcal{C})\})=\max\left\{\;\varPhi(C)\;\middle|\;
C\in\mathcal{C}\setminus \{h(\mathcal{C})\}\;\right\},
\end{equation}
therefore, the restricted evaluation of the clustering $\mathcal{C}$ is:
\begin{equation}\label{clustering_measure}
\Phi_{S}(\mathcal{C})=\max\left\{\varPhi_{S}(h(\mathcal{C})),\;\Phi(\mathcal{C}\setminus\{h(\mathcal{C})\})\right\}.
\end{equation}

%Equations (\ref{head_measure}), (\ref{below_head_measure}) and (\ref{clustering_measure}) are the key to use dynamic programming.

Let $S$ be a subtree of $T$, and let $\cltring{H}{l,S,\mu}$ denote the set of $l$-clusterings of $S$ in which $\mu$ is the weight of the lightest edge in the head cluster. That is:
$$
\cltring{H}{l,S,\mu} = \left\{\;\mathcal{C}\;\mid\;\mathcal{C}\in\cltring{P}{l,S}\text{\quad and\quad}\mu=\min(E(h(\mathcal{C})))\;\right\}.
$$

We are now ready to state an encoding of a local solution and the invariant that allows us to apply dynamic programming:
%properties for Problem~\ref{opt_problem}.
\begin{notation}\label{opt_matrix}
	Suppose $\cltring{H}{l,S,\mu}$ is not empty, then a clustering $\mathcal{C}$ in $\cltring{H}{l,S,\mu}$ is encoded by the ordered pair $O_S(l,\mu)=(M,b)$, if the following properties are fulfilled:
    \begin{enumerate}
 \item  $M=\max(Out_S(h(\mathcal{C})))\text{\; and\; }b=\Phi_S(\mathcal{C}).$
	\item $\Phi_S(\mathcal{C})=\min\left\{\;\Phi_S(\mathcal{C'})\;\mid\;\mathcal{C'}\in\cltring{H}{l,S,\mu}\right\}.$
	\item $\max(Out_S(h(\mathcal{C})))=\min\left\{\;\max(Out_S(h(\mathcal{C'})))\;\mid\; \mathcal{C'}\in\cltring{H}{l,S,\mu}\text{\; and\;\;} \Phi_S(\mathcal{C'})=\Phi_S(\mathcal{C})\;\right\}.$
    \end{enumerate}
% 	We encode $\mathcal{C}$ by the ordered pair $O_S(l,\mu)=(M,b)$ where: $$M=\max(Out_S(h(\mathcal{C})))\text{\; and\; }b=\Phi_S(\mathcal{C}).$$
	
	If $\cltring{H}{l,S,\mu}$ is empty, then $O_S(l,\mu)=\badcltr$, where $\infty$ indicates the ``infinity'' value.
%Let $O_S(l,\mu)$ be an ordered pair representing a clustering $\mathcal{C}\in\cltring{H}{l,S,\mu}$ in the following format:
%$$O_S(l,\mu)=(M,b),\text{\quad where}$$
%$$b=\Phi_S(\mathcal{C})=\min\left\{\;\Phi_S(\mathcal{C'})\;\mid\;\mathcal{C'}\in\cltring{H}{l,S,\mu}\;\right\},\text{\quad and}$$
%$$M=\min\left\{\;\max(Out_S(h(\mathcal{C})))\;\mid\; \mathcal{C}\in\cltring{H}{l,S,\mu}\text{\; and\;\;} \Phi_S(\mathcal{C})=b\;\right\}.$$
\end{notation}

By Lemma~\ref{lemma1}, an optimal clustering $\mathcal{C}^*$ of $T$ has an evaluation $\Phi(\mathcal{C}^*)\leq 1$ and according to Lemma~\ref{division_lemma}, if a clustering $\mathcal{C}$ of a subtree $S$ is used to build $\mathcal{C}^*$, then $\Phi_S(\mathcal{C})\leq 1$ too. Therefore, we set $O_S(l,\mu)$ as $\badcltr$ if $1<\min\left\{\;\Phi_S(\mathcal{C})\;\mid\;\mathcal{C}\in\cltring{H}{l,S,\mu}\;\right\}$. Then, given a subtree $S$, $O_S(\cdot,\cdot)$ is a function whose domain is $\mathbb{N}_{[1,k]}\times (w(E)\cup \{1\})$ and image $\{\badcltr\}\cup (w(E)\cup\{0\})\times \mathbb{R}_{[0,1]}$ where $w(E)=\{\;w(e)\;\mid\; e\in E\;\}$. 

\begin{remark}\label{O_S as table}
Some times, it is more convenient to see $O_S(\cdot,\cdot)$ as a table of $k$ rows with labels $1,2,\dots,k$ and $n$ columns with labels $w(e_1),w(e_2),\dots,w(e_{n-1}),1$ where $e_1, e_2,\dots,e_{n-1}$ is a labeling of the edges in $E$ from the lightest one to the heaviest one. In this way $O_S(l,\mu)$ refers to the cell with row-label $l$ and column-label $\mu$ and it is value is the corresponding ordered pair $(M,b)$.
\end{remark}

If $O_S(l,\mu)=(M,b)\neq\badcltr$, then, by using equations (\ref{head_measure}), (\ref{headless_measure}) and (\ref{clustering_measure}), we obtain that $O_S(l,\mu)$ encodes a clustering $\mathcal{C}$ (not necessarily unique) where: \begin{equation}\label{b_formule}
\Phi_S(\mathcal{C})=b=\max\left\{\frac{M}{\mu},\Phi(\mathcal{C}\setminus\{h(\mathcal{C})\})\right\}.
\end{equation}

For the sake of simplicity, we use the following notation for $O_S(l,\mu)=(M,b)$ (not necessarily distinct from $\badcltr$): $O_S(l,\mu)[1]=M\text{\quad and \quad}O_S(l,\mu)[2]=b.$

If we have the function $O_T$, then the evaluation of the optimal clusterings for Problems \ref{opt_problem} and \ref{opt_problem_k} are: $$\min\left\{\;O_T(k,\mu)[2]\;\mid\;\mu\in w(E)\cup\{1\}\;\right\},\text{\; and}$$
$$\min\left\{\;O_T(l,\mu)[2]\;\mid\;l\in\mathbb{N}_{[2,n]}\text{\; and \;\;}\mu\in w(E)\cup\{1\}\;\right\},\text{\; respectively}.$$

The following lemma is a useful technical result:
\begin{lemma}\label{monotony}
Let $S$ be a subtree rooted at $v$. Let $\mathcal{C}$ and $\mathcal{C'}$ be two different clusterings of $S$ such that $\min(E(h(\mathcal{C})))=\min(E(h(\mathcal{C}')))=\mu$. If $\Phi_S(\mathcal{C})<\Phi_S(\mathcal{C'})\leq 1$ then $\max(Out_S(h(\mathcal{C})))\leq \max(Out_S(h(\mathcal{C'})))$.
\end{lemma}
\begin{proof}
	\begin{figure}
\centering
\begin{subfigure}{.32\textwidth}
\centering
\includegraphics[page = 2]{monotony1.pdf}
\caption{}
\label{monotony1-1}
\end{subfigure}
\begin{subfigure}{.32\textwidth}
	\centering
	\includegraphics[page = 1]{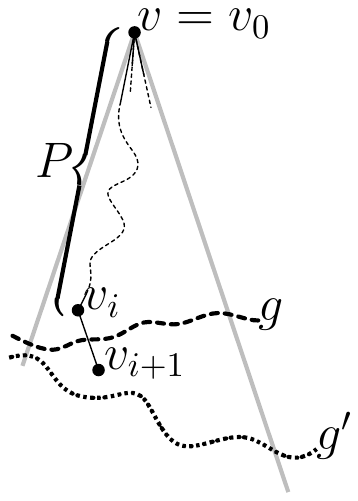}
	\caption{}
	\label{monotony1-2}
\end{subfigure}
\begin{subfigure}{.32\textwidth}
	\centering
	\includegraphics[page = 3]{monotony1.pdf}
	\caption{}
	\label{monotony1-3}
\end{subfigure}
\caption{Let $\mathcal{C}$ and $\mathcal{C'}$ be two different clusterings in a subtree $S$. The head clusters $\head{C}$ and $\head{C'}$ are formed by the nodes above $g$ and $g'$, respectively. Let $P=(v=v_0,v_1,\dots,v_i)$ be a path in $S$. In the three pictures $e=\{v_i,v_{i+1}\}\in Out_S(\head{C})$, note that $e$ is stabbed by $g$. (a) $e\in Out_S(\head{C'})$, note that $e$ is stabbed by $g'$. (b) $e\in \head{C'}$, note that $v_i$ and $v_{i+1}$ are both above $g'$. (c) Some edge in $P$ is in $Out_S(\head{C'})$, note that $P$ is stabbed by $g'$.}
\label{monotony1}
\end{figure}
For the sake of contradiction, suppose that $\max(Out_S(h(\mathcal{C})))> \max(Out_S(h(\mathcal{C'})))$. Observe that $\mu\geq \max(Out_S(h(\mathcal{C})))> \max(Out_S(h(\mathcal{C'})))$, since $\Phi_S(\mathcal{C})<\Phi_S(\mathcal{C'})\leq 1$. Note that $\max(Out_S(h(\mathcal{C})))> \max(Out_S(h(\mathcal{C'})))$ implies that $\max(Out_S(h(\mathcal{C})))>0$ and then $Out_S(h(\mathcal{C}))$ is not empty. Let $e$ be one of the heaviest edges of $Out_S(h(\mathcal{C}))$. Let $P=(v=v_0,v_1,\dots,v_i)$ denote a path from $v$ to $v_{i}$ such that $e=\{v_i,v_{i+1}\}$ (see Figure~\ref{monotony1}). 

If none of the edges of $P$ are in $Out_S(h(\mathcal{C}'))$, then $e\in Out_S(h(\mathcal{C'}))$ or $e\in E(h(\mathcal{C}'))$.
\begin{itemize}
\item[] If $e\in Out_S(h(\mathcal{C'}))$ (Figure~\ref{monotony1-1}) then there is a contradiction, because $$w(e)=\max(Out_S(h(\mathcal{C})))> \max(Out_S(h(\mathcal{C'}))).$$
\item[] If $e\in E(h(\mathcal{C'}))$ (Figure~\ref{monotony1-2}), using $\min(E(h(\mathcal{C}')))=\mu$ yields $\mu\leq w(e)$ and using $\mu\geq \max(Out_S(h(\mathcal{C})))$ leads to $\mu\geq w(e)$, so, $\mu=w(e)$. Therefore, there is a contradiction because: $$1=\frac{w(e)}{\mu}=\varPhi_S(h(\mathcal{C}))\leq \Phi_S(\mathcal{C})<\Phi_S(\mathcal{C}')\leq 1.$$
\end{itemize}

Suppose that an edge $e'\in P$ is in $Out_S(h(\mathcal{C}'))$ (Figure~\ref{monotony1-3}). Note that $w(e')\leq \max(Out_S(h(\mathcal{C}')))$. Also, note that $e'$ is in $E(h(\mathcal{C}))$, and consequently, $\mu\leq w(e')$. Therefore $\mu\leq w(e')\leq \max(Out_S(h(\mathcal{C}')))$ which is a contradiction since $\mu\geq \max(Out_S(h(\mathcal{C})))> \max(Out_S(h(\mathcal{C'})))$.
\end{proof}

From the previous lemma, the following result is deduced directly:
\begin{corollary}\label{minmb}
Let $S$ be a subtree of $T$. For a given value $O_S(l,\mu)=(M,b)\neq\badcltr$, every $l$-clustering $\mathcal{C}\in\cltring{H}{l,S,\mu}$ fulfills that:
$\Phi_S(\mathcal{C})\geq b$, and $\max(Out_S(h(\mathcal{C})))\geq M$.
\end{corollary}

The following lemma is the key of the proposed dynamic programming:
\begin{lemma}\label{opt_struct}
Let $S$ be a subtree rooted at $v$. Let $O_S(l,\mu)=(M,b)\neq\badcltr$ and let $\mathcal{C}$ be an $l$-clustering of $S$ encoded by $O_S(l,\mu)$. Let $Q$ be a subtree of $S$ rooted at $v'\in c(v)$. By removing the edge $e=\{v',v\}$ from $S$ an $l'$-clustering $\mathcal{A}$ of $Q$ is induced. Let $\mu'=\min(E(h(\mathcal{A})))$. By replacing $\mathcal{A}$ with a clustering $\mathcal{B}$ encoded by $O_{Q}(l',\mu')$ and restoring the edge $e$, a new clustering $\mathcal{C}'$ of $S$ is obtained, which is also encoded as $O_S(l,\mu)=(M,b)$ (see Figures \ref{key_e_out} and \ref{key_e_in}).
\end{lemma}
\begin{proof}

By using Corollary~\ref{minmb}, we have that
\begin{equation}\label{observation1}
\Phi_Q(\mathcal{B})\leq\Phi_Q(\mathcal{A}),\;\text{ and\; }\max(Out_Q(h(\mathcal{B})))\leq \max(Out_Q(h(\mathcal{A}))).
\end{equation} 
By using Lemma~\ref{division_lemma}, we have that $\Phi_Q(\mathcal{A})\leq b$, then:
\begin{equation}\label{observation2}
b\geq \Phi_Q(\mathcal{A})\geq\varPhi_Q(h(\mathcal{A}))=\frac{\max(Out_Q(h(\mathcal{A})))}{\mu'}
\end{equation}
Moreover, using observations (\ref{observation1}) and (\ref{observation2}) yields to:
\begin{equation}\label{observation3}
\Phi_Q(\mathcal{B})\leq b.
\end{equation}

Obviously, $\mathcal{C}$ and $\mathcal{C}'$ are both $l$-clusterings. We need to prove that: $$\min(E(h(\mathcal{C})))=\min(E(h(\mathcal{C}')))=\mu,$$ $$\Phi_S(\mathcal{C})=\Phi_S(\mathcal{C}')=b,\text{\; and\; }$$ $$\max(Out_S(h(\mathcal{C})))=\max(Out_S(h(\mathcal{C}')))=M.$$ We divide the rest of the proof into two parts according to the two possible situations when $e$ is going to be removed, $e\in Out_S(h(\mathcal{C}))$ (see Figure~\ref{key_e_out}) or $e\in E(h(\mathcal{C}))$ (see Figure~\ref{key_e_in}). 

Let us start with the first case. Note that, by replacing $\mathcal{A}$ by $\mathcal{B}$, the head cluster is not affected. Consequently, $\max(Out_S(h(\mathcal{C})))=\max(Out_S(h(\mathcal{C}')))=M$, and $\min(E(h(\mathcal{C})))=\min(E(h(\mathcal{C}')))=\mu$. Let us prove that $\Phi_S(\mathcal{C})=\Phi_S(\mathcal{C}')=b$. By Corollary~\ref{minmb}, it is enough to prove that $\varPhi_S(C)\leq b$  for every cluster $C\in\mathcal{C'}$. For every cluster $C\in\mathcal{C'}$ such that $C\not\in \mathcal{B}$, we have that $C$ is also contained in $\mathcal{C}$, so, $\varPhi_S(C)\leq b$. For every cluster $C\in\mathcal{B}$, such that $C\neq h(\mathcal{B})$ we have $\varPhi_S(C)=\varPhi_Q(C)\leq \Phi_Q(\mathcal{B})\leq b$ by observation (\ref{observation3}). From observation (\ref{observation1}) we can deduce that $\varPhi_S(h(\mathcal{B}))\leq \varPhi_S(h(\mathcal{A}))\leq b$.

Let us analyze the second case. Let $\mathcal{D}$ denote the induced clustering of the (remaining) subtree rooted at $v$ (see Figure~\ref{key_e_in}). Note that: $$\min(E(h(\mathcal{C'})))=\min\left\{\min(E(h(\mathcal{D}))),w(e),\min(E(h(\mathcal{B})))\right\},\text{\; and}$$ 
$$\min(E(h(\mathcal{C})))=\min\left\{\min(E(h(\mathcal{D}))),w(e),\min(E(h(\mathcal{A})))\right\}.$$
Notice that $\min(E(h(\mathcal{A})))=\min(E(h(\mathcal{B})))=\mu'$ by construction, therefore, $\min(E(h(\mathcal{C'})))=\min(E(h(\mathcal{C})))$. By Corollary~\ref{minmb}, it is enough to prove that $\max(Out_S(h(\mathcal{C'})))\leq M$ and $\varPhi_S(C)\leq b$ for every cluster $C\in\mathcal{C'}$. Note that:
$$\max(Out_S(h(\mathcal{C'}))) = \max\left\{\max(Out_S(h(\mathcal{D}))), \max(Out_S(h(\mathcal{B})))\right\},\text{\; and}$$
$$\max(Out_S(h(\mathcal{C})))= \max\left\{\max(Out_S(h(\mathcal{D}))), \max(Out_S(h(\mathcal{A})))\right\}.$$
Since $\max(Out_S(h(\mathcal{B})))\leq \max(Out_S(h(\mathcal{A})))$, therefore, $$\max(Out_S(h(\mathcal{C'})))\leq \max(Out_S(h(\mathcal{C})))=M.$$
Finally, for every cluster $C\in\mathcal{C'}\setminus\{h(\mathcal{C'})\}$ such that $C\not\in \mathcal{B}$, $C$ is also contained in $\mathcal{C}$, so, $\varPhi_S(C)\leq b$. For every cluster $C\in(\mathcal{C'}\setminus\{h(\mathcal{C'})\})\cap\mathcal{B}$ we have $\varPhi_S(C)=\varPhi_Q(C)\leq \Phi_Q(\mathcal{B})\leq b$ by observation (\ref{observation3}). From observation (\ref{observation1}) we can deduce that $\varPhi_S(h(\mathcal{C'}))\leq \varPhi_S(h(\mathcal{C}))\leq b$.
\end{proof}

\begin{figure}[t]
	\centering
	\begin{subfigure}{.24\textwidth}
		\centering
		\includegraphics[page=5,scale=0.7]{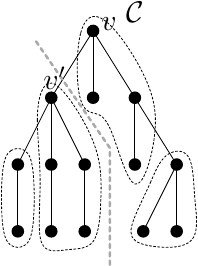}
		\caption{}
	\end{subfigure}
	\begin{subfigure}{.24\textwidth}
		\centering
		\includegraphics[page=6,scale=0.7]{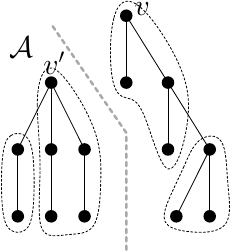}
		\caption{}
	\end{subfigure}
	\begin{subfigure}{.24\textwidth}
		\centering
		\includegraphics[page=7,scale=0.7]{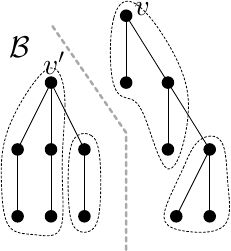}
		\caption{}
	\end{subfigure}
	\begin{subfigure}{.24\textwidth}
		\centering
		\includegraphics[page=8,scale=0.7]{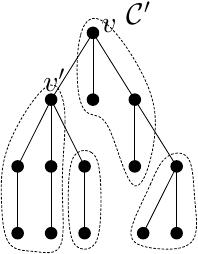}
		\caption{}
	\end{subfigure}
	\caption{Removing the edge $e=\{v,v'\}$ when $e$ connects nodes in different clusters. (a) Initial situation. (b) Induced clustering $\mathcal{A}$ when $e$ is removed. (c) Replacing $\mathcal{A}$ with another clustering $\mathcal{B}$. (d) Restoring the edge $e$ and obtaining a new clustering $\mathcal{C}'$.}
	\label{key_e_out}
\end{figure}
\begin{figure}[t]
	\centering
	\begin{subfigure}{.24\textwidth}
		\centering
		\includegraphics[page=1,scale=0.7]{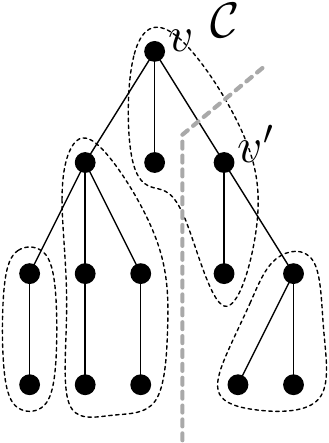}
		\caption{}
	\end{subfigure}
	\begin{subfigure}{.24\textwidth}
		\centering
		\includegraphics[page=2,scale=0.7]{key_dynamic.pdf}
		\caption{}
	\end{subfigure}
	\begin{subfigure}{.24\textwidth}
		\centering
		\includegraphics[page=3,scale=0.7]{key_dynamic.pdf}
		\caption{}
	\end{subfigure}
	\begin{subfigure}{.24\textwidth}
		\centering
		\includegraphics[page=4,scale=0.7]{key_dynamic.pdf}
		\caption{}
	\end{subfigure}
	\caption{Removing the edge $e=\{v,v'\}$ when $e$ is inside a cluster. (a) Initial situation. (b) Induced clustering $\mathcal{A}$ when $e$ is removed. (c) Replacing $\mathcal{A}$ with another clustering $\mathcal{B}$. (d) Restoring the edge $e$ and obtaining a new clustering $\mathcal{C}'$.}
	\label{key_e_in}
\end{figure}

In the next subsections we show how to perform the operations \textbf{UpToParent} and \textbf{AddChildTree}. %analyze two basic ``operations'' of our algorithm:  
%\begin{itemize}
%\item \textbf{UpToParent}: computes $O_{\overline{S}}$ from $O_S$, where root of $S$ is $v$, and $\overline{S}$ is the subtree formed by $S$ and $p(v)$, see Figure~\ref{up2parent}.
%
%\item \textbf{AddChildTree}: computes $O_{P}$ from $O_S$ and $O_Q$, where the root of $Q$ is $v$ and the root of $S$ is $p(v)$, and $P$ is the resultant subtree of joining $S$ and $Q$, see Figure~\ref{merging}.
%\end{itemize} 

In order to simplify the formulas in the next subsections we introduce the following total order: 
% comparison among the elements in $\{\badcltr\}\cup (w(E)\cup\{0\}\times\mathbb{R}_{[0,1]})$ because we need it in the next sections.
Let $(M,b)$ and $(M',b')$ be two ordered pairs. We say that $(M,b)=(M',b')$ if $M=M'$ and $b=b'$. We say that $(M,b)<(M',b')$ if $b<b'$, or if $b=b'$ and $M<M'$.

 \begin{figure}
\centering
\begin{subfigure}{.45\textwidth}
\centering
\includegraphics[page =7, height=4cm]{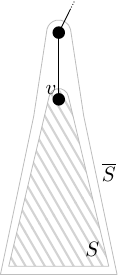}
\caption{}
\label{up2parent}
\end{subfigure}
\begin{subfigure}{.45\textwidth}
	\centering
	\includegraphics[page =8, , height=4cm ]{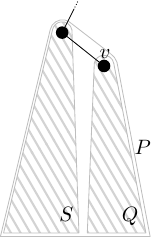}
	\caption{}
	\label{merging}
\end{subfigure}
\caption{(a) Subtree $\overline{S}$ formed by the subtree $S=T_v$ and $p(v)$. (b) Subtree $P$ formed by joining the subtree $Q=T_v$ and a subtree $S$ rooted at $p(v)$.}
 \end{figure}
 
 \subsection{\textbf{UpToParent}: computing $O_{\overline{S}}$ from $O_S$}\label{S2Sp}

  \begin{figure}[t]
  	\centering
  	\begin{subfigure}{.45\textwidth}
  		\centering
  		\includegraphics[page=3]{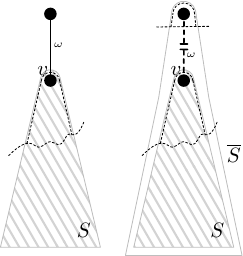}
  		\caption{}
  		\label{up2parent_cut}
  	\end{subfigure}
  	\begin{subfigure}{.45\textwidth}
  		\centering
  		\includegraphics[page=4]{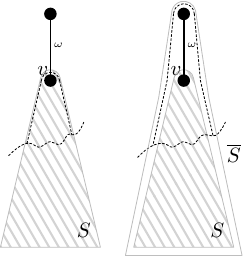}
  		\caption{}
        \label{up2parent_Ncut}
  	\end{subfigure}
  	\caption{Construction of a clustering of $\overline{S}$ based in one of $S$. (a) The edge $\{v,p(v)\}$ is cut. (b) The edge $\{v,p(v)\}$ is not cut.}
  \end{figure}
  
 Let $S=T_v\neq T$ be a subtree and let $\overline{S}$ denote the tree formed by the union of $S$ and $p(v)$. In this section we will show how to compute $O_{\overline{S}}$, assuming that we already know $O_S$. Let $\omega=w(\{v,p(v)\})$. If we are computing $O_{\overline{S}}(l,\mu)$, then:
 
 \begin{claim}\label{S2Sp-res1}
 	If $\mu=1$, then:
$$O_{\overline{S}}(l,\mu)=\left(\omega, \min_{\mu'}\left\{\max\left\{\omega,O_S(l-1,\mu')[2],\frac{\max\{\omega,O_S(l-1,\mu')[1]\}}{\mu'}\right\}\right\}\right)$$
 \end{claim}
 \begin{proof}
Let $\mathcal{C}$ be a clustering encoded as $O_{\overline{S}}(l,\mu)=(M,b)$. We have that $h(\mathcal{C})=\{p(v)\}$ because $\mu=1$.
 	Then, the edge $\{v,p(v)\}$ is cut and $M=\omega$, see Figure~\ref{up2parent_cut}. Using (\ref{b_formule}) leads to $b=\max\{\omega, \Phi(\mathcal{C}\setminus\{h(\mathcal{C})\})\}$. Note that $\mathcal{C}\setminus\{h(\mathcal{C})\}$ is an $(l-1)$-clustering $\mathcal{C'}$ of $S$. 
    According to Lemma~\ref{opt_struct}, $\mathcal{C'}$ is encoded as $O_S(l-1,\mu')$ for some $\mu'$, and
    then:
    $$\Phi_S(\mathcal{C'})=O_S(l-1,\mu')[2]\text{\; and \; }\varPhi(h(\mathcal{C'}))=\frac{\max\{\omega,O_S(l-1,\mu')[1]\}}{\mu'}.$$
    
    It is easy to see that:
$$\Phi(\mathcal{C'})=\max\left\{\varPhi(h(\mathcal{C'})),\Phi(\mathcal{C'}\setminus \{h(\mathcal{C'})\})\right\},$$
by using equation (\ref{headless_measure}) we have:
$$\Phi(\mathcal{C'})=\max\left\{\varPhi(h(\mathcal{C'})),\Phi_S(\mathcal{C'}\setminus \{h(\mathcal{C'})\})\right\}.$$
It is also easy to see that $\varPhi(h(\mathcal{C'}))\geq \varPhi_S(h(\mathcal{C'}))$, therefore:
$$\Phi(\mathcal{C'})=\max\left\{\varPhi(h(\mathcal{C'})),\varPhi_S(h(\mathcal{C'})),\Phi_S(\mathcal{C'}\setminus \{h(\mathcal{C'})\})\right\}.$$
The previous equation can be reduced to $\Phi(\mathcal{C'})=\max\left\{\varPhi(h(\mathcal{C'})),\Phi_S(\mathcal{C'})\right\}.$

\begin{eqnarray*}
\text{Finally,\; }\Phi(\mathcal{C}\setminus\{h(\mathcal{C})\})=\Phi(\mathcal{C'})&=&\max\left\{\Phi_S(\mathcal{C'}),\varPhi(h(\mathcal{C'}))\right\}\\
&=&\max\left\{O_S(l-1,\mu')[2],\frac{\max\{\omega,O_S(l-1,\mu')[1]\}}{\mu'}\right\}.
\end{eqnarray*}
 	
 \end{proof}
 
\begin{claim}\label{S2Sp-res2}
If $1>\mu>\omega$ then $O_{\overline{S}}(l,\mu)=\badcltr$.
\end{claim}
\begin{proof}
It is impossible to build a clustering with this encoding. If we cut $\{v,p(v)\}$, then $\mu=1$, and if we do not cut $\{v,p(v)\}$, then it is in the head cluster and then $\mu\leq \omega$.
\end{proof}

\begin{claim}\label{S2Sp-res3}
If $\mu=\omega$ then:
$$O_{\overline{S}}(l,\mu)=\min_{\mu'\geq\omega}\left\{\left(O_S(l,\mu')[1],\max\left\{\frac{O_S(l,\mu')[1]}{\omega},O_S(l,\mu')[2]\right\}\right)\right\}.$$
\end{claim}
\begin{proof}
Let $\mathcal{C}$ be a clustering encoded as $O_{\overline{S}}(l,\mu)=(M,b)$. In this case $\mu\leq\omega$, so, $\{v,p(v)\}$ is not cut, see Figure~\ref{up2parent_Ncut}. According to Lemma~\ref{opt_struct}, $\mathcal{C}$ is formed by adding $p(v)$ to the head cluster of an $l$-clustering $\mathcal{C'}$ of $S$, which is encoded as $O_S(l,\mu')$ for some $\mu'\geq\omega$. Note that $\max(Out_{\overline{S}}(h(\mathcal{C})))=\max(Out_S(h(\mathcal{C'})))$ and then $M=O_S(l,\mu')[1]$. 
It is easy to see that $\Phi_{\overline{S}}(\mathcal{C}\setminus\{h(\mathcal{C})\})=\Phi_{S}(\mathcal{C}'\setminus\{h(\mathcal{C}')\})$ and $\varPhi_{\overline{S}}(h(\mathcal{C}))\geq \varPhi_{S}(h(\mathcal{C}'))$, so:
\begin{eqnarray*}
b=\Phi_{\overline{S}}(\mathcal{C})&=&\max\left\{\varPhi_{\overline{S}}(h(\mathcal{C})),\Phi_{\overline{S}}(\mathcal{C}\setminus\{h(\mathcal{C})\})\right\}\\
&=&\max\left\{\varPhi_{\overline{S}}(h(\mathcal{C})),\varPhi_{S}(h(\mathcal{C}')),\Phi_{S}(\mathcal{C}'\setminus\{h(\mathcal{C}')\})\right\}\\
&=&\max\left\{\varPhi_{\overline{S}}(h(\mathcal{C})),\Phi_S(\mathcal{C}')\right\}.
\end{eqnarray*}
The result follows.
\end{proof}

\begin{claim}\label{S2Sp-res4}
If $\mu<\omega$, then:
$$O_{\overline{S}}(l,\mu)=O_S(l,\mu).$$
\end{claim}
\begin{proof}
Let $\mathcal{C}$ be a clustering encoded as $O_{\overline{S}}(l,\mu)=(M,b)$.
In this case $\mu\leq\omega$, and consequently $\{v,p(v)\}$ is not cut, see Figure~\ref{up2parent_Ncut}. According to Lemma~\ref{opt_struct}, $\mathcal{C}$ is formed by adding $p(v)$ to the head cluster of an $l$-clustering $\mathcal{C'}$ of $S$ which is encoded as $O_S(l,\mu)$. It is easy to see that $\max(Out_{\overline{S}}(h(\mathcal{C})))=\max(Out_{S}(h(\mathcal{C}')))=M$ and $\Phi_{\overline{S}}(\mathcal{C})=\Phi_S(\mathcal{C'})=b$. The result follows.
\end{proof}

\begin{theorem}
Let $S=T_v$ such that $T_v\neq T$, and let $\overline{S}$ be the subtree formed by adding $p(v)$ to $S$.
If we know the function $O_S$, then the function $O_{\overline{S}}$ can be computed in $O(kn)$ time.
\end{theorem}
\begin{proof}
Think in $O_S(\cdot,\cdot)$ as table (see Remark~\ref{O_S as table}), lets analyze the time to compute the values of every cell in this table.
By Claim~\ref{S2Sp-res1}, computing the values of form $O_{\overline{S}}(l,1)$ takes $O(n)$ time per cell and there are $O(k)$ cells of this form, resulting in the total time of $O(kn)$. By Claim~\ref{S2Sp-res2}, the values of form $O_{\overline{S}}(l,\mu)$ with $1>\mu>\omega$ take constant time per cell and given that there are $O(kn)$ cells of this form, the total time results to $O(kn)$. 
By Claim~\ref{S2Sp-res3}, computing the values of the form $O_{\overline{S}}(l,\omega)$ takes $O(n)$ time per cell and there are $O(k)$ cells of this form, yielding a total time of $O(kn)$. Finally, by Claim~\ref{S2Sp-res4}, computing the values of form $O_{\overline{S}}(l,\mu)$ with $\mu<\omega$ takes constant time per cell and there are $O(kn)$ cells of this form, resulting in a total time of $O(kn)$. The result follows.
\end{proof}

 \subsection{\textbf{AddChildTree}: computing $O_P$ from $O_S$ and $O_Q$}\label{SQ2P}
\begin{figure}[t]
	\centering
	\begin{subfigure}{.45\textwidth}
		\centering
		\includegraphics[page=5]{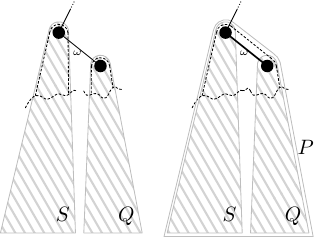}
		\caption{}
		\label{merging_nocut}
	\end{subfigure}
	\begin{subfigure}{.45\textwidth}
		\centering
		\includegraphics[page=6]{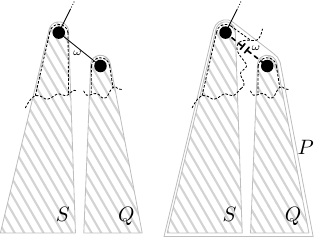}
		\caption{}
		\label{merging_cut}
	\end{subfigure}
	\caption{Construction of a clustering of $\overline{S}$ based in one of $S$. (a) The edge $\{v,p(v)\}$ is cut. (b) The edge $\{v,p(v)\}$ is not cut.}
\end{figure}

Let $Q=T_v\neq T$ and let $S$ be a subtree rooted in $p(v)$ such that $S$ does not contain $v$. Let $P$ denote the subtree which results from joining on $S$ and $Q$. We show how to compute $O_P$ from $O_S$ and $O_Q$. Let $\omega=w(\{v,p(v)\})$. Let $\overline{Q}$ be the subtree formed by adding $p(v)$ to $Q$. Pre-compute $O_{\overline{Q}}$ from $O_Q$ using the claims of the previous subsection. If we are computing $O_{P}(l,\mu)$, then we have that:

\begin{claim}\label{add_res1}
If $\mu=1$ then:
$$O_P(l,\mu)=\min_{1\leq x<l}\left\{\left(\max\left\{\omega, O_S(x,1)[1]\right\},\max\left\{O_S(x,1)[2],O_{\overline{Q}}(l-x+1,1)[2]\right\}\right)\right\}$$
\end{claim}
\begin{proof}
Let $\mathcal{C}$ be a clustering encoded as $O_{P}(l,\mu)=(M,b)$. We have that $h(\mathcal{C})=\{p(v)\}$, since $\mu=1$.
 	Then, the edge $\{v,p(v)\}$ is cut, see Figure~\ref{merging_cut}. $\mathcal{C}$ is based in two clusterings $\mathcal{C}'$ and $\mathcal{C}''$ with encodings $O_S(x,1)$ and $O_{\overline{Q}}(l-x+1,1)$ for some $1\leq x<l$. It easy to see that:
    \begin{eqnarray*}
    M=\max(Out_P(h(\mathcal{C})))&=&\max\left\{\omega,\max(Out_S(h(\mathcal{C'})))\right\}\\
    &=&\max\left\{\omega,O_S(x,1)[1]\right\}.
    \end{eqnarray*}
    It is easy to see that:
    $$b=\Phi_P(\mathcal{C})=\max\left\{\varPhi_P(h(\mathcal{C})),\Phi_P(\headless{C})\right\}.$$
    Note also that: $$\varPhi_P(\head{C})=\max\left\{\varPhi_S(\head{C'}),\varPhi_{\overline{Q}}(\head{C''})\right\},\text{\; and}$$
    $$\Phi_P(\headless{C})=\max\left\{\Phi_S(\headless{C'}),\Phi_{\overline{Q}}(\headless{C''})\right\}.$$
    Then,
    $$\Phi_P(\mathcal{C})=\max\left\{\varPhi_S(\head{C'}),\varPhi_{\overline{Q}}(\head{C''}),\Phi_S(\headless{C'}),\Phi_{\overline{Q}}(\headless{C''})\right\}.$$
    And this can be rewritten as:
    
    \begin{align*}
    \Phi_P(\mathcal{C})&=\max\left\{\Phi_S(\mathcal{C'}),\Phi_{\overline{Q}}(\mathcal{C''})\right\}\\
    &=\max\left\{O_S(x,1)[2],O_{\overline{Q}}(l-x+1,1)[2]\right\}.
    \end{align*}    
\end{proof}

\begin{claim}\label{add_res2}
If $1>\mu>\omega$ then:
$$O_P(l,\mu)=\min_{1\leq x<l}\left\{\left(M,\max\left\{\frac{M}{\mu},O_S(x,\mu)[2],O_{\overline{Q}}(l-x+1,1)[2]\right\}\right)\right\},$$
where $M=\max\left\{\omega, O_S(x,\mu)[1]\right\}$.
\end{claim}
\begin{proof}
Let $\mathcal{C}$ be a clustering encoded as $O_{P}(l,\mu)=(M,b)$. Note that $\{v,p(v)\}$ is not included in $\head{C}$ because $\mu>\omega$, and consequently this edge is cut, see Figure~\ref{merging_cut}. Furthermore, $\mathcal{C}$ is based in two clusterings, $\mathcal{C}'$ and $\mathcal{C}''$, with encodings $O_S(x,\mu)$ and $O_{\overline{Q}}(l-x+1,1)$ for some $1\leq x<l$. It is easy to see that $M=\max\left\{\omega, O_S(x,\mu)[1]\right\}$.
Notice:
$$\varPhi_P(\head{C})\geq \varPhi_S(\head{C'}),\text{\; and\; }\varPhi_P(\head{C})\geq \varPhi_{\overline{Q}}(\head{C''}).$$
Moreover,
$$\Phi_P(\headless{C})=\max\left\{\Phi_S(\headless{C'}),\Phi_{\overline{Q}}(\headless{C''})\right\}.$$
Now, focus on the evaluation of $\mathcal{C}$:
\begin{align*}
b=\Phi_P(\mathcal{C})&=\max\left\{\varPhi_P(\head{C}),\Phi_P(\headless{C})\right\}\\
&=\max\left\{\varPhi_P(\head{C}),\varPhi_S(\head{C'}),\varPhi_{\overline{Q}}(\head{C'}),\Phi_S(\headless{C'}),\Phi_{\overline{Q}}(\headless{C''})\right\}\\
&=\max\left\{\frac{M}{\mu},\Phi_S(\mathcal{C'}),\Phi_{\overline{S}}(\mathcal{C''})\right\}\\
&=\max\left\{\frac{M}{\mu},O_S(x,\mu)[2],O_{\overline{Q}}(l-x+1,1)[2]\right\}.
\end{align*}
\end{proof}
\begin{claim}\label{add_res3}
If $\mu=\omega$ then:
\begin{align*} 
o_1&=\min_{1\leq x<l}\left\{\left(M_1,\max\left\{\frac{M_1}{\omega},O_S(x,\omega)[2],O_{\overline{Q}}(l-x+1,1)[2]\right\}\right)\right\},\\
&\text{where \;}M_1=\max\left\{\omega, O_S(x,\omega)[1]\right\};\\
o_2&=\min_{1\leq x<l,\;\;\omega\leq\mu'}\left\{\left(M_2,\max\left\{\frac{M_2}{\omega},O_S(x,\mu')[2],O_{\overline{Q}}(l-x+1,\omega)[2]\right\}\right)\right\},\\
&\text{where\; }M_2=\max\left\{O_S(x,\mu')[1], O_{\overline{Q}}(l-x+1,\omega)[1]\right\};\text{\; and}\\
O_P(l,\mu)&=\min\{o_1,o_2\}.
\end{align*}
\end{claim}
\begin{proof}
Let $\mathcal{C}$ be a clustering encoded as $O_{P}(l,\mu)=(M,b)$. In this case, we have two options to build $\mathcal{C}$. The first one is using two clusterings, $\mathcal{C'}$ and $\mathcal{C''}$, with encodings $O_S(x,\omega)$ and $O_{\overline{Q}}(l-x+1,1)$, respectively. In this case, $\{v,p(v)\}$ is cut. This case is analogous to the previous claim and corresponds to $o_1$ encoding. And the second one, using two clusterings $\mathcal{C'}$ and $\mathcal{C''}$ with encodings $O_S(x,\mu')$ for some $\mu'\geq \omega$ and $O_{\overline{Q}}(l-x+1,\omega)$, respectively. In this case, $\{v,p(v)\}$ is not cut. This case corresponds to $o_2$ encoding and we can prove it using ideas similar to the ones used in the previous claims.
\end{proof}

\begin{claim}\label{add_res4}
If $\mu<\omega$ then:
\begin{align*} 
o_1&=\min_{1\leq x<l,\;\;\mu\leq\mu'}\left\{\left(M_1,\max\left\{\frac{M_1}{\mu},O_S(x,\mu)[2],O_{\overline{Q}}(l-x+1,\mu')[2]\right\}\right)\right\},\\
&\text{where\; }M_1=\max\left\{O_S(x,\mu)[1], O_{\overline{Q}}(l-x+1,\mu')[1]\right\};\\
o_2&=\min_{1\leq x<l,\;\;\mu\leq\mu'}\left\{\left(M_2,\max\left\{\frac{M_2}{\mu},O_S(x,\mu')[2],O_{\overline{Q}}(l-x+1,\mu)[2]\right\}\right)\right\},\\
&\text{where\; }M_2=\max\left\{O_S(x,\mu')[1], O_{\overline{Q}}(l-x+1,\mu)[1]\right\};\\\text{\; and}\\
O_P(l,\mu)&=\min\{o_1,o_2\}.
\end{align*}
\end{claim}
\begin{proof}
Let $\mathcal{C}$ be a clustering encoded as $O_{P}(l,\mu)=(M,b)$. In this case $\mu<\omega$ and then $\{v,p(v)\}$ is not cut (otherwise, the head cluster of $\mathcal{C}$ has an evaluation greater than 1). $\mathcal{C}$ is based in clusterings $\mathcal{C'}$ and $\mathcal{C''}$ of $S$ and $\overline{Q}$, respectively. There are two possible ways to build $\mathcal{C}$: the lightest edge into $\head{C}$ is in $\head{C'}$, $o_1$; or the lightest edge into $\head{C}$ is in $\head{C''}$, $o_2$. In both cases, the formulas can be verified using the same ideas used in the previous claims.
\end{proof}

\begin{theorem}
Let $Q=T_v$ such that $T_v\neq T$, and let $S$ be a subtree rooted at $p(v)$ such that $v$ is not in $S$. Let $P$ denote the subtree formed by joining $S$ and $Q$.
If we know the functions $O_S$ and $O_Q$, then the function $O_{P}$ can be computed in $O(k^2n^2)$.
\end{theorem}
\begin{proof}
Analyzing the number of cells in each claim (\ref{add_res1}, \ref{add_res2}, \ref{add_res3}, \ref{add_res4}) and the complexity to compute the value of a cell in each case, we conclude that $O_{P}$ can be computed in $O(k^2n^2)$.
\end{proof}

\subsection{Complexity of the algorithm}
Given a tree $T$ and a value $k$, we can calculate $O_T$ by computing $O_{T_v}$ for every node $v$ in $T$ in a bottom-up (from the leaves to the root) procedure using the mentioned operations. Note that, if $v$ is a leaf, then $O_{T_v}(1,1)=(0,0)$ and $O_{T_v}(l,\mu)=\badcltr$ if $l>1$ or $\mu<1$. To compute the function $O_{T_v}$ of an inner node $v$, we proceed as follows: Let $\{v_1,\dots,v_m\}$ be the set of children of $v$. First, considering $S=T_{v_1}$, compute $O_{\overline{S}}$ from $O_S$ using the \textbf{UpToParent} operation. Subsequently, we proceed with joining the subtrees $T_{v_i}$ one by one using the \textbf{AddChildTree} operation. When all the children have been added, the resulting subtree corresponds to $T_v$. Note that we apply a single operation per edge. Consequently, this algorithm takes $O(k^2n^3)$ time. Note also that with this algorithm, we obtain the evaluation of the optimal clustering; the clusters of an optimal solution can be computed by ``\emph{navigating backwards}'' through the computed functions.

Problem~\ref{opt_problem_k} can be solved using the same idea with a slightly more complex approach. We can use a similar algorithm based on functions $O_{S}(\mu)$, saving the parameter $l$, (which corresponds to the number of clusters) and then the computing time is $O\left(n^3\right)$.

%\subparagraph*{Acknowledgments.} We thank the organizers for
%the tasty cookies.
\bibliographystyle{abbrv}
\bibliography{clustering}

\end{document}